%% file: main.tex
\documentclass{article}

\usepackage[english]{babel}
\usepackage[blocks]{authblk}

\usepackage[style=numeric-comp, sorting=none]{biblatex}
\usepackage{graphicx}
\usepackage{svg}
\graphicspath{ {./images/} }

\usepackage{sgame}

\usepackage[a4paper, margin=1in]{geometry}
\usepackage{subcaption}

\usepackage{amsmath,amsthm}
\usepackage{tabularx}
\usepackage{booktabs,siunitx,multirow}
\usepackage{float}
\usepackage{graphicx}
\usepackage{amssymb}
\usepackage{epigraph}
\usepackage[para]{threeparttable}
\usepackage{adjustbox}
\usepackage[colorlinks=true,linkcolor=black, citecolor=blue, urlcolor=black]{hyperref}
\usepackage[labelfont=bf, labelsep=period, labelformat=default]{caption}
\usepackage{cleveref}
\crefname{equation}{Eq.}{Eqs.}
\creflabelformat{equation}{#2#1#3}

\usepackage[toc=true, sort=use]{glossaries}
\makeglossaries

\newtheorem{theorem}{Theorem}[section]

\newcommand{\authorblock}[3]{%
  \begin{minipage}[t][4.5em][s]{0.30\linewidth}
    \centering
    \vspace{0pt}
    \textbf{#1} 
    \vfill
    #2 
    \vfill
    \url{#3} 
  \end{minipage}
}

\AtBeginBibliography{\small}

\newcommand{\lucas}[1]{\textcolor{teal}{\textbf{LB: }\textit{#1}}}

\usepackage{float, soul}

\usepackage{url}
\title{\textbf{A New Incentive Model For Content Trust}}

\author{%
  \vspace{2em}
  \small
  \centering
  \authorblock{Lucas Barbosa}{PwC Australia}{lucas.chu.barbosa@au.pwc.com}
  \hfill
  \authorblock{Sam Kirshner}{UNSW}{s.kirshner@unsw.edu.au}
  \hfill
  \authorblock{Rob Kopel}{PwC Australia}{rob.kopel@au.pwc.com}
  \vspace{2em}

  \makebox[\textwidth][c]{%
    \authorblock{Eric Tze Kuan Lim}{UNSW}{e.t.lim@unsw.edu.au}
    \authorblock{Tom Pagram}{PwC Australia}{tom.pagram@au.pwc.com}
  }
  \vspace{1em}
}

\date{}

\begin{document}
\input{gloss}

\begin{titlepage}

\maketitle
\thispagestyle{empty}

\begin{abstract}
\noindent This paper outlines an incentive-driven and decentralized approach to verifying the veracity of digital content at scale. Widespread misinformation, an explosion in AI-generated content and reduced reliance on traditional news sources demands a new approach for content authenticity and truth-seeking that is fit for a modern, digital world. By using smart contracts and digital identity to incorporate ``trust’’ into the reward function for published content, not just engagement, we believe that it could be possible to foster a self-propelling paradigm shift to combat misinformation through a community-based governance model. The approach described in this paper requires that content creators stake financial collateral on factual claims for an impartial jury to vet with a financial reward for contribution. We hypothesize that with the right financial 
and social incentive model users will be motivated to participate in crowdsourced fact-checking and content creators will place more care in their attestations. This is an exploratory paper and there are a number of open issues and questions that warrant further analysis and exploration.
\end{abstract}

\end{titlepage}

\pagenumbering{roman}
\setcounter{page}{1} 

\newpage
\tableofcontents
\newpage

\pagenumbering{arabic}

\section{Introduction}

\subsection{Background and Motivation}

The Internet and social networks have profoundly reshaped how individuals communicate, access economic opportunities, and engage with information. However, the decentralized nature of online content has also contributed to widespread misinformation and a reduced reliance on traditional news sources and government \cite{whototrust}. This erosion of confidence is compounded by the sheer volume of digital content, which makes discerning truth from falsehood increasingly challenging.

Recent advances in large language models (LLMs) further complicate this landscape. State-of-the-art LLMs have shown remarkable capability for persuasive discourse, often surpassing human performance in debate contexts \cite{rescala2024}, particularly when they have access to personal data that can be used to tailor arguments \cite{salvi2025}. As these models grow more sophisticated and accessible, the scale and quality of \gls{synthetic content} are projected to increase significantly \cite{ai-growth:2018, merging-with-machine:2010}. Indeed, experts estimate that AI systems could generate up to 90\% of online content by 2026 \cite{deepfakes}.

In this evolving context, malicious actors will likely exploit automated content-generation systems to spread false or misleading information. Without robust countermeasures, distinguishing authentic information from manipulated content will become progressively difficult \cite{adapting-to-llms:2023, societal-impact-of-llms:2021}, as AI-driven misinformation can circulate faster than mainstream news \cite{shao:2017}. Given the velocity and volume of digital content, traditional fact-checking approaches may struggle to keep pace \cite{pace-of-fake-news:2017}, underscoring the pressing need for new strategies to verify both the integrity and the provenance of online information.

\subsection{The Problem Statement}

Alongside these emerging challenges, modern media comprises two interconnected substructures: traditional (institutional) media and social media. Audiences have shown reduced reliance on traditional news sources \cite{whototrust}, partly due to perceived gaps in journalistic transparency and the deterioration of in-depth reporting \cite{trust-decline-journalism:2019}. High-profile instances, such as coverage of the alleged weapons of mass destruction in Iraq (2003) \cite{wmds} and the misrepresentation of factors contributing to the 2008 Global Financial Crisis \cite{the-big-short}, have eroded confidence in mainstream outlets. Consequently, trust between the public and institutions including governments stands at an unprecedented low \cite{trust-decline-institution:2017, edelman-trust-report:2023}.

Meanwhile, social media platforms inherently reward content that drives engagement. Because misinformation typically attracts more immediate attention and can be produced quickly at lower cost than carefully researched, fact-based material \cite{mitSpread2018}, there is a structural incentive to prioritize engagement over accuracy. Moreover, users often face minimal repercussions for spreading misinformation, which can foster confirmation bias and reinforce polarized ``echo chambers'' \cite{sunstein:echo-chambers, pariser:echo-chambers}, thereby distorting public perceptions of truth \cite{social-media-echo-chamber}.

Collectively, these conditions reveal a disjointed information ecosystem in which neither traditional nor social media systems are consistently aligned with truthfulness. The key challenge lies in establishing a robust and scalable framework that can effectively incentivize accuracy, transparency, and accountability in content creation and dissemination. Addressing this challenge is essential for restoring public trust and safeguarding the integrity of information in the digital age.

\subsection{Overview of the Proposed Framework}


We propose a set of incentive-driven design principles for addressing \gls{misinformation}: to establish trust in authentic and trustworthy content by allowing 1) content \gls{creators} to have skin in the game connected to the truthfulness of their content, 2) \gls{challengers} to \gls{contest} the veracity of that content, which also requires them to have skin in the game, and 3) \gls{jurors}, with financial and social incentives for high-quality efforts, to adjudicate these contents. This proposed framework will leverage modern standards \cite{c2pa} to provide jurors with detailed provenance information for digital assets, such as images and videos, to provide insight into the origins and integrity of evidence. All actors in this ecosystem (i.e., the content creators, challengers and jurors) will have verified digital identities. Each contest could be managed through transparent, immutable, and auditable smart contracts executed and paid out via a distributed ledger.

Specifically, the incentive framework is driven by the financial collateral that content creators are permitted to stake with their content as a \gls{veracity bond}\footnote{Charles Hoskinson, founder of Cardano, was an early advocate for the concept of a ``veracity bond'' \cite{charles-veracity-bond:2024}.} (VB), signaling confidence in the veracity of the content. If readers believe the content is inaccurate, they can become a challenger, providing evidence to substantiate their claim, and crucially, they must stake a \gls{counter-veracity bond} (CVB) of equal value the VB. The jury’s role is to determine whether the challenger's contest is valid or if the creator's content is true. The collateral of the losing party is then distributed to the winning parties as compensation.

Existing approaches using crowd-sourced annotations \cite{community-notes, birdwatch} have proven successful \cite{annotation-effectiveness}. However, despite their use, crowd-sourced annotations present a vulnerability as large groups with non-diverse political viewpoints may still propagate misinformation \cite{diversity-in-skeptical-views, bellingcat-taylor-swift}. Labeling and annotation of content is an important part of fact checking \cite{effectiveness-of-warning-labels}, however, the ideal approach should go beyond reactive measures, focusing instead on proactively addressing the incentives behind content creation \cite{annotation-effectiveness}. It has been demonstrated that readers can successfully discern high-quality and reputable content outlets from misleading ones \cite{pennycook-crowdsourced-solution} and that veracity bonds increase perceived credibility \cite{barbosa2025}. We suggest that harnessing the collective intelligence of the crowds with monetary compensation \cite{monetary-incentives} supported by game theory \cite{game-theoretic-incentives:2016, game-theoretic-short-long-term:2020, game-theoretic-large-stakes:2015, game-theoretic-mobile-networks:2012} will produce a net effect greater on the spread of misinformation online than crowd-sourced annotations alone.







\section{Rules and Mechanics}

\subsection{Roles and Responsibilities}

\subsubsection{Creators}

A fundamental challenge in information dissemination is incentivizing the creation and consumption of accurate content. One possible approach involves requiring creators to stake veracity bonds with a principal value of $\beta$ as a guarantee of the accuracy of their content, which is explored further in \S\ref{sec:content_and_bonds}. The underlying premise is that a substantial audience values reliable information, generating demand for creators that prioritize truthfulness.  

The financial risk associated with staking veracity bonds is primarily limited to the opportunity cost of having the principal locked during the staking period. However, this cost may be offset by the increased visibility and engagement afforded to content backed by significant bonds. Presumably, content with larger veracity bonds receives preferential treatment, such as greater visibility and prominence within the community, which can enhance the creator's reputation and standing in their field.  

In addition, a challenge mechanism may be incorporated, enabling users to dispute the veracity of staked content. The opportunity for challengers to profit from successful challenges or for creators to benefit from unsuccessful ones could provide an additional incentive for user participation, fostering an ecosystem where both creators and consumers actively contribute to information verification. By aligning incentives for accuracy and engagement, this framework aims to promote a more reliable and critically engaged information environment.

\subsubsection{Challengers} 

Challengers can dispute the accuracy of a creator’s content by compiling and presenting evidence to support their claim. To initiate a challenge, a challenger must stake a counter-veracity bond with the same principal value $\beta$ as the veracity bond of the creator. While multiple challenges may be brought against a single piece of content, each challenger is limited to one challenge per content. These challenges are processed sequentially in a randomly determined order (see \S\ref{sec:turn_based_game_structure} for further details).  

Challengers are responsible for providing sufficient evidence to demonstrate that the content is inaccurate. If the jury rules in favor of the challenger, they retain their bond and receive a portion of the creator’s bond, determined by the payout function $\pi_c$. Conversely, if the challenge is unsuccessful, $\pi_c$ is awarded to the creator. In both cases, the jurors are compensated according to the payout function $\pi_j$.



\subsubsection{Jurors}

A collective of independently selected jurors\footnote{To avoid conflicts of interest, neither the challenger nor creator may serve as a juror of their own contest.}, each verified through a robust digital identity process that adheres to the standards outlined in \S\ref{sec:digital-identity}, is responsible for reviewing and deliberating contest. The compensation of the juror, denoted by $\pi_j$, is a fraction of $\beta$ (see \cref{eq:beta_breakdown}), regardless of whether it is sourced from  the creator or challenger and should reflect the quality of the evaluation of the juror. 

Each juror is required to show active participation, including thorough examination of all evidence, clear justification of the perspectives, and casting votes that align with the available facts. Such actions ensure that compensation is tied to the substance and consistency of each juror’s evaluation, rather than superficial involvement. Although a logarithmic scale is one way to ensure that excessively high bonds do not overshadow smaller ones to maintain an approximate equal distribution of juror effort, the actual form of $\pi_j$ can be adapted to different contexts. In any implementation, $\pi_j$ should encourage broad participation, ensuring that jurors deliberate across a diverse set of contests, and reward quality and consistency, so that well-supported assessments lead to more favorable outcomes.

Every selected juror must post a refundable bond of $\gamma\beta$ ($0 < \gamma < 1$, with $\gamma$ serving only as an adjustable constant) before deliberation begins. The bond is returned in full when the juror (i) attends and votes within the deliberation window, (ii) submits a written assessment of the evidence, and (iii) later receives an average rating of at least ``neutral'' on the yes / neutral / no scale used in the anonymous juror-evaluation process described in \S\ref{sec:juror-eval}. If the juror is absent, abstains, or is predominantly rated ``no'', the bond is forfeited to a protocol reserve fund rather than being redistributed under \cref{eq:beta_breakdown}. A candidate who does not provide the bond is replaced immediately by an alternate juror who can. The bond neither affects selection probability nor reputation scoring; it simply extends financial accountability to the jury while preventing any transfer of value to disputing parties.



\subsubsection{Viewers}

Viewers are users who consume content without actively participating as creators, challengers, or jurors\footnote{Creators, challengers and jurors may also be viewers however for simplicity, this paper assumed these roles to be distinct.}. Their role is crucial for driving engagement and shaping the overall value of the ecosystem. By facilitating access to truthful content, viewers are empowered to make well-informed decisions and engage more meaningfully with the information they consume.

\subsection{Rules Governing Interactions}
\subsubsection{Content Creation and Veracity Bonds}
\label{sec:content_and_bonds}

A veracity bond is a form of financial collateral that creators can place on their content to demonstrate confidence in its accuracy \cite{barbosa2025}. By depositing a substantial bond, creators signal a strong belief in the veracity of their work. Conversely, withholding a bond may reflect lower confidence or an inability to provide immediate proof of accuracy \cite{swedish-journalists:2019}. As illustrated in \autoref{fig:sys-creator}, creators can deposit a veracity bond immediately after writing the content.

\begin{figure}[h]
    \centering
    \includegraphics[width=0.4\textwidth]{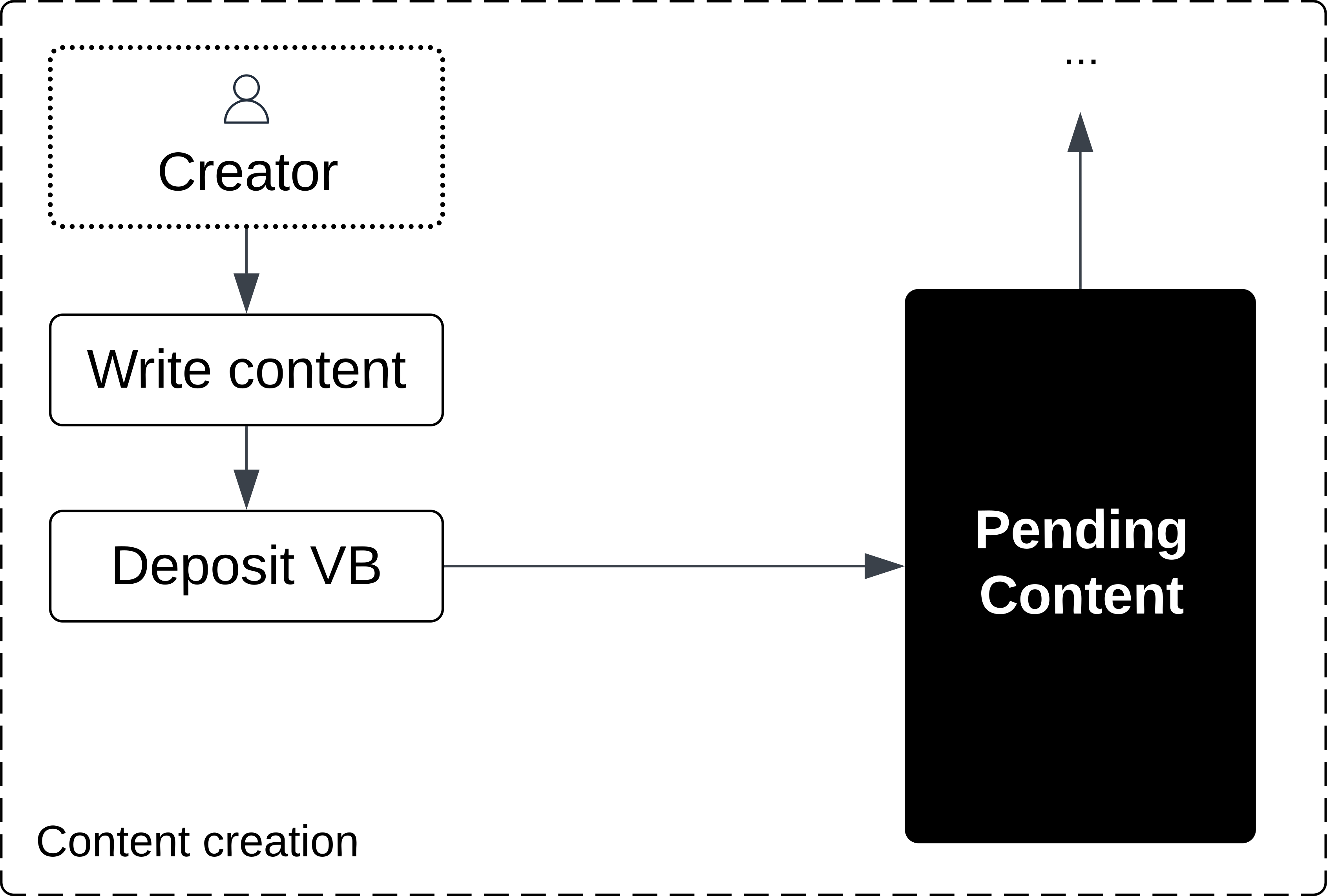}
    \caption{Creators can deposit a veracity bond immediately after writing the content.}
    \label{fig:sys-creator}
\end{figure}

If the content is found to be misleading, the bond is forfeited and redistributed to the challenger who contested and the jurors who deliberated. This forfeiture occurs only if the jury verdict rules against the creator, which treats the bonded content as a formal legal contract. Recipients should receive payment through a blockchain-based smart contract. This smart contract ensures secure and immutable transfers of funds, reinforcing trust by eliminating the need for a centralized authority \cite{nakamoto2008,eth-paper,trust-in-smart-contracts}. If no successful disputes arise during the \gls{challenge period}, the original creator is refunded the value of their bond $\beta$ in full.

Not all creators have the resources to conduct extensive fact-checking, particularly freelance journalists whose work may not undergo verification by large organizations \cite{freelance-journalists:2016}. Therefore, any creator with a verified identity is eligible to deposit a bond on their content. Encouraging participation from a diverse range of creators enhances the overall information ecosystem by expanding the availability of verified news and perspectives online \cite{diverse-perspectives}.

\subsubsection{Challenging Content and Counter Veracity Bonds}

Challenging content requires the challenger to post a counter-veracity bond that matches the value of the original veracity bond $\beta$ as illustrated in \autoref{fig:sys-challenger}. The equal stake ensures symmetrical incentives for both parties, preventing the jury from being skewed in a direction that prioritizes monetary gain over the factual integrity of content. By requiring equal bonds from creators and challengers, the approach aims to maximize fairness and impartiality.

\begin{figure}[h]
    \centering
    \includegraphics[width=0.4\textwidth]{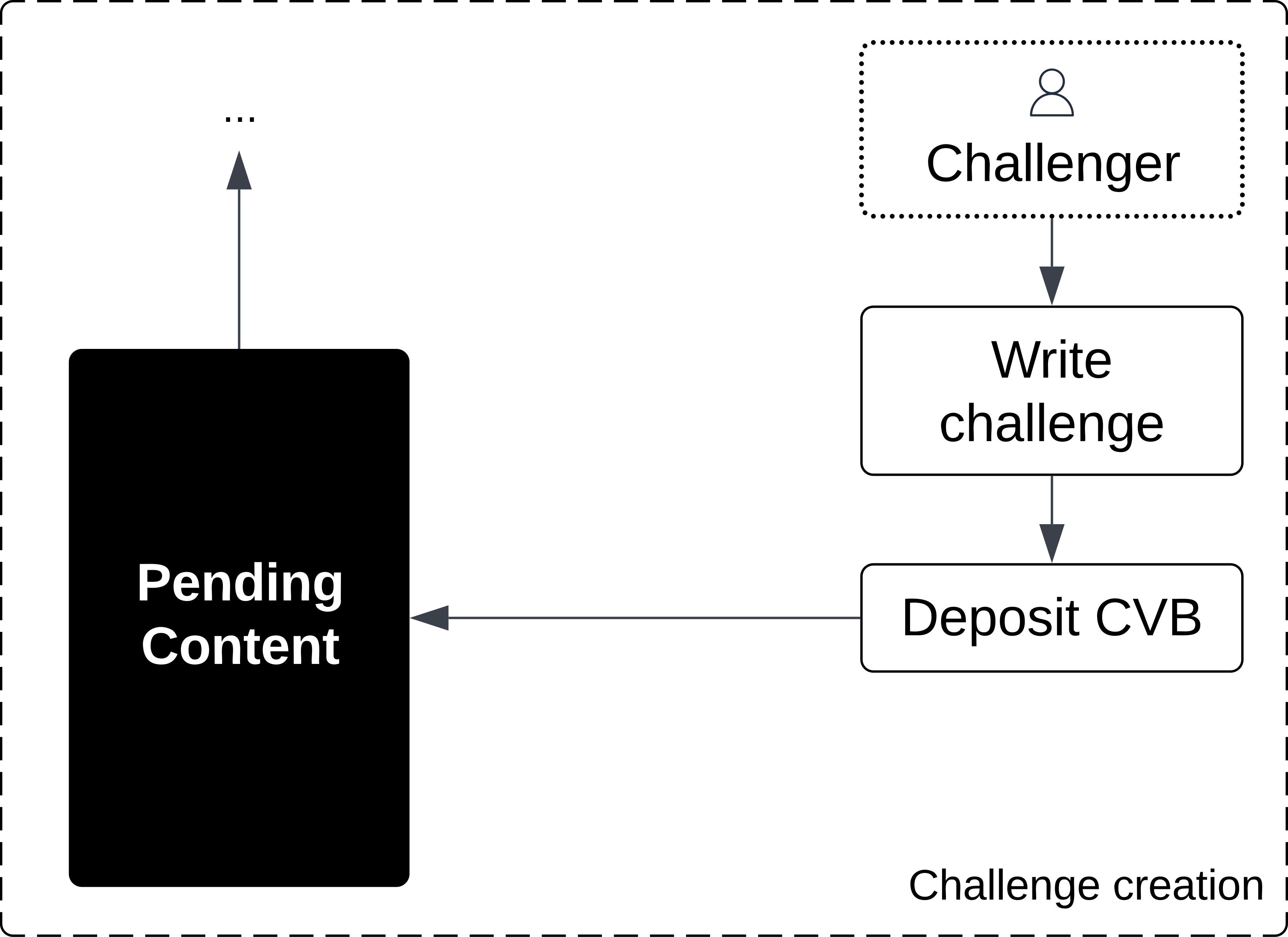}
    \caption{Challengers can deposit a counter-veracity bond immediately after writing the challenge.}
    \label{fig:sys-challenger}
\end{figure}

The requirement for a counter-veracity bond also discourages frivolous or malicious challenges. Those initiating a challenge must be sufficiently confident in their contest to risk losing the bond. This mechanism promotes thorough fact-checking and reduces exploitation attempts. In addition, each challenger can submit only one challenge per distinct unit of content. If a challenger loses, they forfeit the CVB, and the next challenger in the queue is randomly picked. This sequence continues until either a challenger succeeds or the challenge period ends.

In extreme cases, challengers with considerable resources might attempt to monopolize the challenge process by submitting multiple challenges across different content, effectively creating a backlog that delays genuine disputes. To address this, safeguards can be implemented, such as limiting the number of active challenges that a single party can initiate at a given time and using a randomized challenger queue (discussed further in \S\ref{sec:turn_based_game_structure}) policies to prevent large-scale manipulation and ensure timely dispute resolution for all parties.

\subsubsection{Jury Resolutions}

Resolutions are a central part of the framework, providing a clear outcome to each dispute while allowing for future reassessment as human knowledge advances. Although resolutions offer closure at the time they are rendered, they remain subject to reconsideration whenever new evidence or insights emerge.\footnote{Consider the evolving understanding of health recommendations regarding fats. Decades ago, dietary guidelines commonly advocated avoiding all forms of fat to maintain heart health. However, with advances in nutritional science, it is now widely accepted that unsaturated fats, such as those found in olive oil and avocados, are beneficial for cardiovascular health \cite{pacheco2022avocado}.}

To reach a resolution, a jury of $n$ members is selected from a pool of $N$ eligible jurors to reach a majority consensus within the \gls{deliberation period}. An odd-numbered jury (where $n = 2m+1 \ll N$ and $n, m, N \in \mathbb N$) avoids the possibility of a 50/50 deadlock. The jury’s verdict either supports the creator, validating the content’s accuracy, or favors the challenger, indicating that the content has been successfully disputed.\footnote{This framework could alternatively use an even-numbered jury, allowing for an indeterminate result in which both the VB and CVB are returned to their respective parties.}

Once the verdict is finalized, the forfeited bond (whether VB or CVB) is distributed to the jurors and the winning party, as shown in \autoref{fig:sys-full-flow}. This process ensures that participants are compensated for their time and diligence in evaluating the dispute.

\begin{figure}[H]
    \centering
    \includegraphics[width=0.8\textwidth]{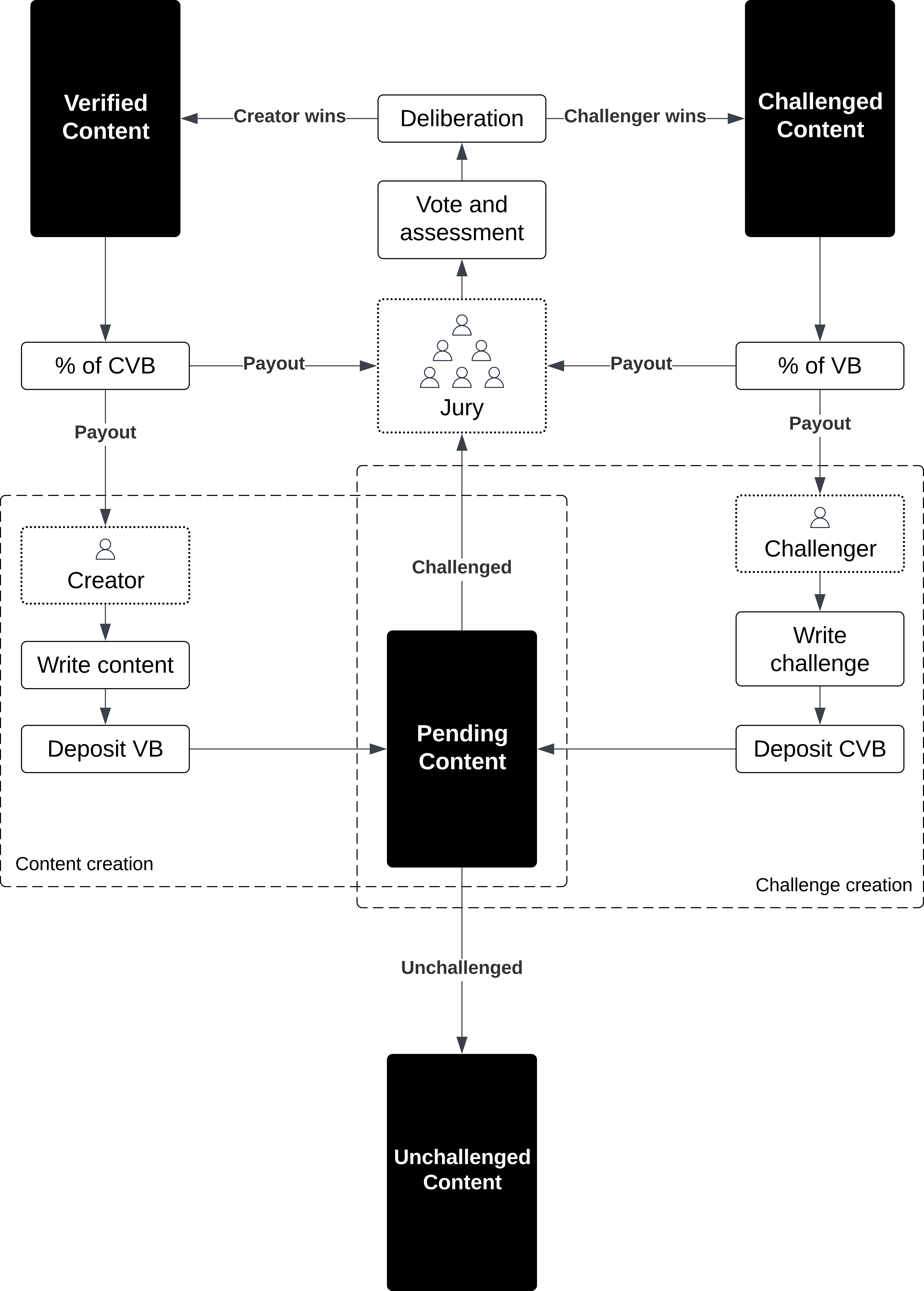}
    \caption{A fixed percentage of the losing bond (VB or CVB) is distributed to the jury and the winning party.}
    \label{fig:sys-full-flow}
\end{figure}

Jurors would ideally participate by reviewing evidence, providing detailed assessments, and casting informed votes within the deliberation period. Should a juror fail to fulfill these responsibilities, they are deemed inactive and must be promptly replaced by the next best available candidate from the pre-established bench of jurors. This inactivity adversely affects their reputation, underscoring the importance of maintaining an excess pool of $N$ qualified jurors. The surplus ensures that any inactive jurors can be seamlessly substituted without compromising the overall quality or timeliness of the deliberation process.

\subsubsection{Juror Evaluations}
\label{sec:juror-eval}

Jurors have considerable influence as they are selected for their subject matter expertise and may receive compensation for resolving disputes. To preserve the integrity of the framework, it is crucial to establish mechanisms that identify and remove jurors who exhibit ineffectiveness or malicious intent. Consequently, a process of random and anonymous evaluations is proposed, allowing select viewers who did not participate in the specific contest to assess a juror’s performance as illustrated in \autoref{fig:sys-juror-eval}. The randomization and anonymity of the evaluations will mitigate the risk of collusion and bias and encourage the evaluators to concentrate on the substance of the juror’s independent assessment.

\begin{figure}[h]
    \centering
    \includegraphics[width=1.0\textwidth]{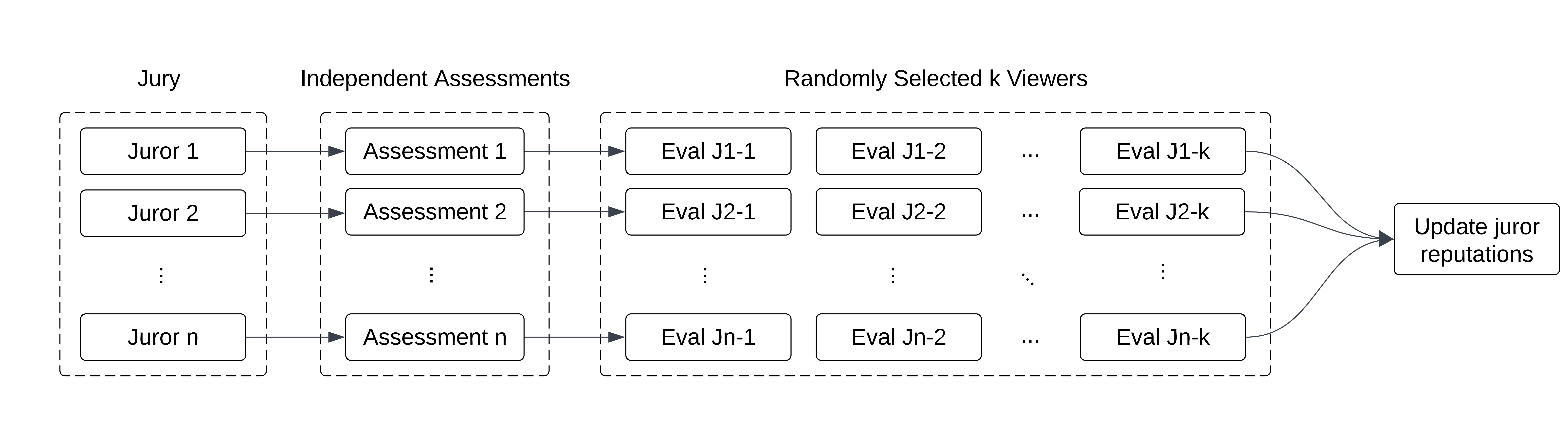}
    \caption{Juror evaluations are randomly assigned to a selected number of viewers who do not have a conflict of interest with that contest.}
    \label{fig:sys-juror-eval}
\end{figure}

Jurors are required to compose detailed assessments explaining their verdicts, outlining the evidence considered, the relative weight assigned to each piece of evidence, and the reasoning behind their conclusions. These comprehensive notes ensure transparency in the jury’s decision-making process. The reviewers can then assess these contributions using a three-point Likert scale to indicate whether the juror's input is ``no'', ``neutral'' or ``yes'' in terms of quality and thoroughness. Jurors who consistently receive lower ratings could experience a decline in reputation, potentially reducing their likelihood of being selected for future juries or, in some cases, leading to disqualification. This framework aims to promote accountability, incentivize rigorous evaluations, and support the delivery of impartial, well-substantiated verdicts across the board.

\subsubsection{Jury Collusion-Resistance}

Collusion among a small subset of jurors poses a potential threat to the integrity of any verdict. Understanding and controlling this risk is essential for any approach that relies on collective judgment. The analysis that follows shows how the likelihood of a coordinated bloc overturning a verdict shrinks quickly as the jury grows, while \S\ref{app:collusion_proof} supplies the complete mathematical proof.

Assume a total pool of $N$ potential jurors, of whom exactly $k$ are willing to coordinate their votes. Each trial draws a jury of $n = 2m + 1$ members; using an odd number eliminates the possibility of a tie. A decision is reached when at least $t = m+1$ jurors vote the same way. Because the jury is sampled without replacement, the number $X$ of coordinating jurors selected follows a $\mathrm{Hypergeometric}(N,k,n)$ distribution, and applying Hoeffding’s inequality \cite{hoeffding1963, chavatal1979} to this variable yields the exponential upper-tail bound,
\begin{equation*}
  P \bigl(X\ge t\bigr)\;\le\;\exp\Bigl[-2n\bigl(\tfrac12-p\bigr)^2\Bigr].
  \label{eq:chernoff}
\end{equation*}

The exponent is negative whenever the required majority $t$ exceeds the expected number of colluders $np$. As $n$ increases, the magnitude of that exponent grows linearly, so the probability of a successful collusion attempt falls exponentially.

To illustrate, suppose $10\%$ of the overall pool is malicious. With a 21 jury the bound above already drops below $0.2\%$. Expanding the jury to 43 members presses the figure below $0.00004\%$. Plots in \autoref{fig:juror-collusion-plots} depict this decline for several choices of $N$; the underlying calculations and tables are given in \S\ref{app:collusion_tables}.

\begin{figure}[h]
    \centering
    \includegraphics[width=1.0\textwidth]{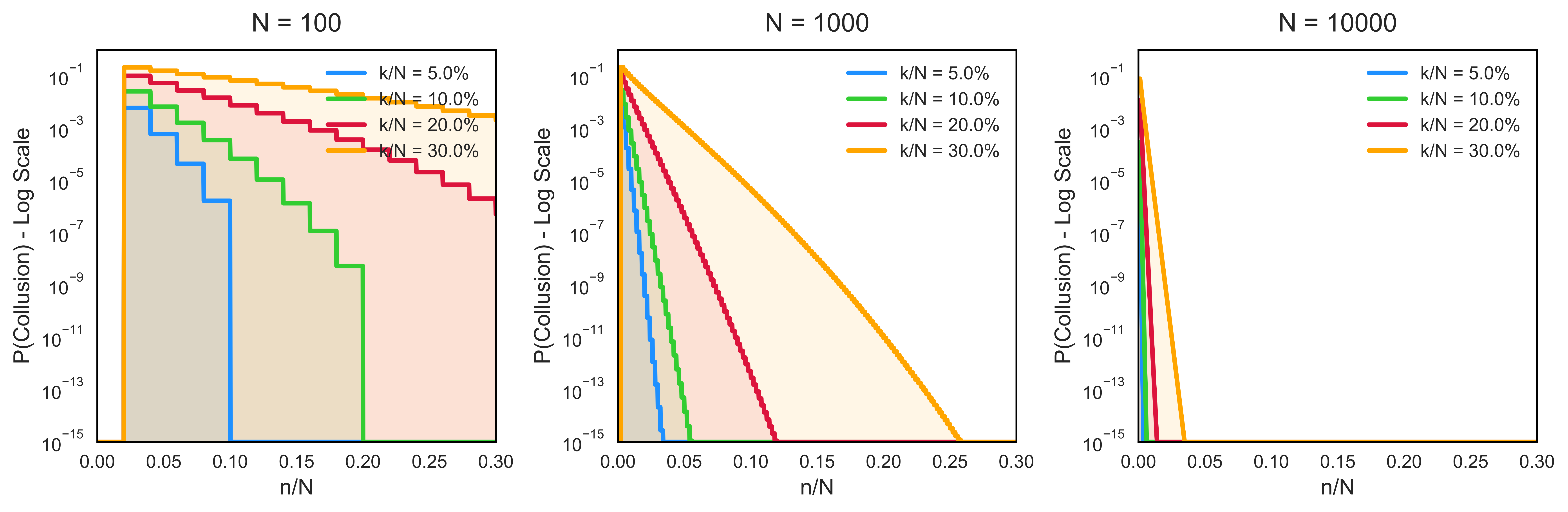}
    \caption{Collusion probability (on a log‐scale) as a function of jury selection ratio ($n/N$) for varying population sizes and corruption levels. The three panels show results for populations of $N=100$, $1\,000$, and $10\,000$ individuals, each with different percentages of unfaithful jurors ($k/N=5\%$, $10\%$, $20\%$, $30\%$).}
    \label{fig:juror-collusion-plots}
\end{figure}

\subsection{Turn-based Game Structure}
\label{sec:turn_based_game_structure}

Once the content is in place and supported by a veracity bond, the framework can be modeled as a turn-based game, which simplifies the application of game-theoretic analysis. The fundamental idea is that each \textit{turn} corresponds to a challenger stepping forward with a counter-veracity bond and initiating a contest over the content’s accuracy. If no challenger emerges within the allotted period (the challenge period), the content is considered unchallenged and the veracity bond is fully returned to the creator, as described below and illustrated in \autoref{fig:sys-full-flow}.

\subsubsection{Sequential Challenge Model}

In an asynchronous challenge model, multiple challengers can simultaneously initiate contests against a creator; however, only one can ultimately succeed. Allowing multiple challengers to win would undermine the economic structure described in \cref{eq:beta_breakdown}. Moreover, managing jury allocations for concurrent contests, while preserving the independence of each proceeding, adds significant complexity.

In contrast, the serial model enforces a more streamlined process by allowing only one challenger at a time to contest the creator, with any additional challengers placed into a randomized queue as shown in \autoref{fig:turn-based}. If the current challenge succeeds, all remaining queued challenges are dismissed and the contest is concluded. If the challenge fails, the next challenger in the queue proceeds. This randomized queue can help prevent resource-rich challengers from monopolizing the challenge process.

\begin{figure}[h]
    \centering
    \includegraphics[width=0.8\textwidth]{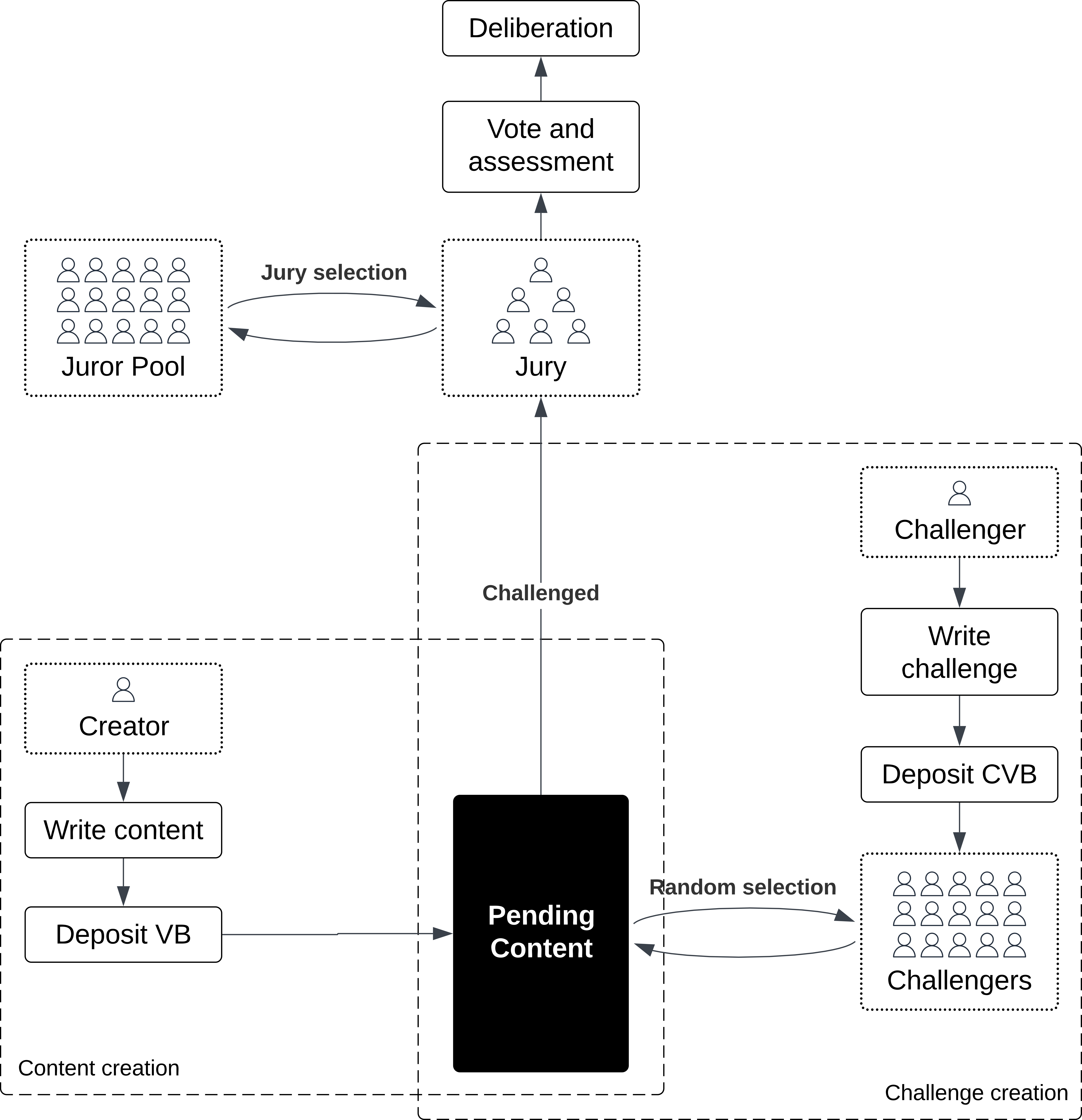}
    \caption{Adapted from \autoref{fig:sys-full-flow} where during game turns, if a challenger loses their contest, a new one will be selected at random. Similarly, if any jurors are inactive during the deliberation period they will be substituted by the next qualified candidate.}
    \label{fig:turn-based}
\end{figure}

\subsubsection{Unchallenged Content}

When the challenge period ends without any contest, the entry is classed as unchallenged and the creator’s veracity bond is returned in full. The absence of a dispute, however, is not a definitive proof of correctness. It can signal at least three distinct situations: (i) the claim is self-evidently accurate so no one sees value in challenging it, (ii) the claim is so trivial or uninteresting that potential challengers do not invest their time, or (iii) the audience that could evaluate the claim has not yet encountered it. Further consideration is needed to determine how each of these possibilities should shape subsequent audits, reputation adjustments, and archival practice.

\section{Incentive Model}

\epigraph{\textit{Show me the incentives, and I'll show you the outcome}.}{Charlie Munger}

\subsection{Payout Structure}

The incentive framework is grounded in the distribution of a forfeited veracity bond or a counter-veracity bond as shown in \autoref{fig:sys-full-flow}. When a contest is finalized, the losing party's $\beta$ is forfeited and redistributed to the winning party, jurors and framework in the following fashion,
\begin{equation}
    \pi_c + \sum_{j \in J} \pi_j + \pi_p = \beta,
    \label{eq:beta_breakdown}
\end{equation}
where $\pi_c$ is the payout awarded to the successful creator or challenger, $\sum_{j \in J} \pi_j$ represents the collective compensation to all jurors who meaningfully participated and $\pi_p$ corresponds to the share of the framework. This closed-loop mechanism underscores the self-sustaining nature of the model: the rewards for accurate assessments are funded exclusively by penalties for inaccurate content, thereby reducing reliance on external subsidies.

From an economic standpoint, this structure aligns the interests of creators, challengers and jurors with the overarching objective of promoting accurate content. Creators who are confident in the veracity of their content are incentivized to stake a veracity bond, potentially realize profits from failed challenges while simultaneously enhancing their reputation through the clear demonstration of trust in their content’s authenticity. Challengers are motivated to dispute potentially misleading material, as they receive a share of the forfeited bond only upon a successful challenge. Jurors may have incentives to issue accurate and high-quality judgments, as erroneous, biased, or low-quality judgments could affect their future opportunities. The framework's fixed percentage of the forfeited bond, $\pi_p$ , finances its operational expenses, such as infrastructure, dispute-resolution protocols, and identity verification. Because the fee is a fixed percentage that is publicly disclosed, everyone can verify that the framework’s total earnings grow only with overall activity, not with hidden fee increases at users’ expense.

A salient feature of this architecture is that the payout structure in \cref{eq:beta_breakdown} does not rely on external funding. Instead, the approach is self-sustaining, with rewards for accurate assessments funded entirely by the penalties imposed on inaccurate content. This closed-loop mechanism reinforces economic incentives while minimizing dependency on external subsidies. The design inherently introduces risk for participants: a creator who publishes a demonstrably true claim and faces no challenge loses nothing, whereas a challenger who mistakenly disputes a true claim stands to lose their counter-veracity bond. Similarly, jurors who make poor-quality biased judgments forfeit potential earnings. These risks are integral to the incentive structure of the framework, ensuring that all parties are motivated to contribute to the overall accuracy and integrity of the content.

\subsection{Reputation}


One significant difficulty in developing a truth-maximizing framework is the imbalance between jurors and challengers. Since voting or providing an independent assessment requires far less effort than preparing a well-researched challenge, the juror pool is likely to outnumber challengers by a wide margin. This gap makes it essential to implement mechanisms that distinguish and reward jurors who deliver rigorous, accurate evaluations, so that low-effort or erroneous judgments do not overwhelm the system. Ultimately, the framework must align incentives to promote high-quality juror contributions, rather than simply maximizing participation

To address this challenge, the framework incorporates an internal reputation mechanism designed to promote prosocial behavior, namely accurate and thorough assessments, while discouraging actions driven solely by self-interest or greed. This reputation mechanism filters the pool of jurors by prioritizing individuals who consistently contribute to accurate content evaluation.

The theoretical foundation for this approach is drawn from the work of \textcite{bnabou2006} on incentives and prosocial behavior, which provides a framework to understand how individuals balance intrinsic motivations, extrinsic monetary rewards, and reputational concerns. The framework models juror behavior by choosing a participation level $a$, drawn from a set of possible actions $A \subset \mathbb{R}$, which can encompass factors such as the thoroughness of the written assessment rated by randomized reviewers, as explained in \S\ref{sec:juror-eval}.

The overall benefit a juror derives from participation can be represented as 
\begin{equation*}
    (v_a + v_y y) a -C(a),
\end{equation*}
where $y$ represents the monetary incentive (that is, the potential payment of the veracity bond), $v_a$ reflects the intrinsic value derived from prosocial participation (such as a sense of civic duty or intellectual satisfaction), $v_y$ indicates the intrinsic value of the monetary reward and $C(a)$ represents the cost associated with the chosen level of participation. A well-designed reputation mechanism effectively increases the perceived cost of low-effort actions, thereby aligning individual incentives with collective accuracy.

Participation should generate a reputational payoff, $R$, which is an essential component of the incentive structure,
\begin{equation}
    R(a,y) = x\left(\gamma_a E(v_a \mid a, y) - \gamma_y E(v_y \mid a, y)\right).
\label{eq:rep_model}
\end{equation}

In \cref{eq:rep_model}, $x$ represents the visibility of the juror's actions within the community, while $\gamma_a$ and $\gamma_y$ capture the relative weights a juror places on prosocial reputation and monetary reward, respectively. The terms $E(v_a \mid a, y)$ and $E(v_y \mid a,y)$ denote the expected intrinsic value of the juror from prosocial participation and the expected valuation of the monetary reward, given their chosen action $a$ and the incentive $y$. A higher $R$ score indicates that a juror is more strongly perceived to be motivated by prosocial concerns than financial gain.

Jurors whose $R$ scores exceed a predefined threshold may be selected, with selection and ranking estimating juror quality while accounting for uncertainty in limited observations. A dynamic ranking mechanism continuously updates evaluations as new data emerge, identifying underperforming jurors and highlighting those who consistently provide high-quality assessments.

In cases where jurors' $R$ scores fall below average, one possibility is that they may not be selected for future juries, creating a link between reputation and reward. By incorporating reputation into the incentive structure and promoting prosocial behavior rather than monetary motivations, this approach could encourage jurors to provide accurate and reliable evaluations, contributing to the overall effectiveness of the community-driven ecosystem.

\subsection{Visibility}

One possible approach is to extend the incentive structure beyond individual participants by linking the visibility of content to the size of its associated veracity bond. Under this design, verified content backed by larger bonds could receive preferential treatment within the ecosystem. For example, the bond size might serve as a weighting factor in ranking algorithms that determine search results, content with larger bonds might display prominent visual indicators to signal its verification status, and recommendation engines could prioritize such content for users with related interests.

The rationale for this approach is that a larger veracity bond can serve as a credible indicator of the creator’s confidence, since the financial risk is higher if the claim is found to be false. This additional risk may motivate creators to engage in more rigorous fact-checking before publication. In turn, larger bonds are likely to attract more focused scrutiny from challengers and jurors, helping ensure that high-visibility content is thoroughly vetted. A potential concern is that this approach could devolve into a ``pay-to-win'' scenario, but in principle, creators are not purchasing visibility; they are staking a bond that can be forfeited if their claims are proven false. The heightened scrutiny by an incentivized pool of challengers and reputable jurors may act as a deterrent against misinformation, regardless of bond size. Furthermore, mechanisms such as the ones described in \S\ref{sec:digital-identity} could be introduced to detect and penalize efforts aimed at artificially inflating bond sizes through collusion or other manipulative behaviors.

This visibility mechanism is designed to function as part of a broader incentive structure, working in tandem with the challenger and juror subsystems to promote accountability and factual accuracy within the community.

\section{Digital Identity}
\label{sec:digital-identity}

A critical vulnerability in any online system that relies on user participation, particularly one involving voting or reputation, is the potential for Sybil attacks \cite{sybil:2002}.  A Sybil attack occurs when a single malicious entity creates multiple fake accounts (often referred to as ``sock puppets'' or ``bots'') to manipulate the system, for example, by casting multiple votes, artificially inflating reputation scores, or flooding the system with dis/misinformation.  In the context of designing a robust truth maximizing framework, preventing Sybil attacks is crucial to maintain the juror pool’s integrity and ensure contest reliability. If malicious actors dominate the juror pool, the entire incentive structure can be undermined.

Achieving Sybil resistance may benefit from a multilayered approach that combines basic verification techniques with advanced digital identity solutions. At the foundational level, measures such as the verification of email and phone numbers, the implementation of CAPTCHA, the monitoring of IP addresses, and the fingerprint of devices serve to deter automated account creation and reduce the risk of fraudulent activity. Although these techniques provide an essential baseline, more robust methods may be needed to counter determined attackers. Enhanced measures could include voluntary government-issued ID verification and linking to established social media accounts, which might offer stronger guarantees of user uniqueness while balancing privacy concerns.

For long-term scalability and privacy preservation, decentralized identity (DID) solutions present a promising avenue to verify the sovereignty of personhood \cite{did}. DIDs, based on the global W3C standard \cite{w3c:did-2022}, enable users to control their own identifiers and associated data, thus reducing dependence on centralized authorities. Requiring new identities to be issued only from a trusted \textit{issuer} and verified by a different trusted entity \textit{verifier} creates substantial friction to prevent malicious automated accounts from undermining the system \cite{microsoft-did}. This use of DIDs remains compatible with the provenance standards \cite{c2pa} explored in \S\ref{sec:provenance} as they are agnostic to identity data as long as it is representable via a W3C Verifiable Credential \cite{w3c:creds-2019}, that includes DIDs. On a global scale, decentralized identifiers offer distinct advantages over government-issued identification, which is generally limited by jurisdiction. Because DIDs adhere to a global standard that ensures verifiability across multiple regions, they are particularly well suited for content platforms serving an international audience.








    

\section{Content Provenance}
\label{sec:provenance}

Digital evidence (e.g., images, audio, video, or files) is essential in providing tangible support for claims about the origin, authenticity, and integrity of content \cite{taxonomy-of-trust}. In this context, content creators and challengers hold the responsibility of providing evidence when appropriate. Creators are encouraged to include evidence with verifiable provenance data to enhance credibility. In turn, challengers should rigorously assess this information and equally provide provenance data to any evidence supporting counterclaims. Jurors then rely on this traceable history of entities, activities, and individuals to evaluate the authenticity and integrity of the evidence presented. Adherence to emerging standards \cite{c2pa, w3c:prov} for embedding provenance metadata in digital assets and promotes truthful reporting and may help deter the distribution of manipulated or fabricated material, such as deep fakes.

\section{Future Work}

\subsection{Challenges and Limitations}

\textbf{Content creator incentives}. While the potential for financial rewards for verified content is attractive, it may not be sufficient to compete with established free platforms. The framework's success depends on active participation in both content creation and the verification process. While this model is not designed to build an entirely new media platform, its true value will be realized only when integrated into existing ecosystems that adopt the protocol. Future work must explore strategies to enhance creator engagement without compromising the model’s integrity.

\textbf{Scalability}. A robust incentive model for content trust relies on a large and engaged community of fact-checkers and jurors. As content volume increases, efficiently managing the verification process becomes more complex. Without effective scalability measures, there is a risk of unverified content slipping through as attention decays in fast-paced news cycles. Furthermore, the model's dependence on user participation exposes it to Sybil attacks, where malicious actors create multiple fake accounts to manipulate the results \cite{sybil:2002}. Implementing strong digital identity verification and refined jury selection mechanisms is critical for sustaining trust at scale.

\textbf{Subjective and evolving nature of truth}. Truth is not static, and new evidence or shifts in societal understanding may necessitate revisiting past verifications. The incentive model must therefore incorporate flexible mechanisms that allow for appeals, re-evaluations, and content updates. Establishing clear guidelines and a transparent process for revising content verification outcomes is essential. This adaptability not only supports ongoing trust but also aligns with the dynamic nature of public discourse.

\textbf{Expert jurors}. Certain content niches require specialized knowledge that only experts can provide. Integrating expert jurors into the verification process can capture nuanced insights that are otherwise overlooked. However, identifying qualified experts and ensuring their selection remains unbiased is a complex challenge. Future work should focus on developing robust verification methods for expert credentials and establishing transparent selection processes. Balancing expert input with community-driven insights will be key to preserving the model’s integrity and trustworthiness.

\textbf{Concentration of powers}. There is a risk that the incentive model may inadvertently favor well-resourced creators, who could more easily afford the cost of producing verified content. This dynamic might exacerbate existing inequalities in information access. To foster a more inclusive ecosystem, further research is needed to find the optimal balance between content quality and accessibility. Policies and incentives should be designed to democratize participation, ensuring that the trust model benefits a diverse range of voices.

\textbf{Bias}. Reliance on community consensus for verification raises concerns about the potential for dominant narratives to suppress alternative viewpoints. Mitigating such biases is crucial for maintaining a fair and balanced trust system. Future research should explore methods such as diversifying the jury pool, weighting votes based on verified user expertise or credibility, and incorporating algorithmic checks. These measures will help safeguard against echo chambers and ensure that the incentives for trust are distributed equitably.

\subsection{Open Questions}

A series of open questions emerge from this exploration. Could the proposed incentive model catalyze a new market that effectively prices truth, thereby incentivizing content creators to produce high-quality, truthful content while providing consumers with reliably verified information? How can the framework be designed to prevent collusion among jurors, ensuring that coordinated voting does not compromise the fairness and objectivity of the verification process? What is the optimal balance between speed and rigor in determining the durations of the challenge and deliberation periods, such that the process remains relevant in fast-moving news cycles while still allowing for a thorough evaluation of complex issues? Additionally, what should be the ideal size of a jury to promote diverse perspectives and efficient decision-making, and what level of anonymity is appropriate for challengers and jurors to maintain impartiality without sacrificing accountability? How can the selection and integration of expert witnesses be optimized to ensure that the most qualified individuals are assigned to relevant cases, while simultaneously mitigating potential biases within the framework? Furthermore, how should the framework handle content that remains unchallenged, ensuring that the absence of a challenge does not unduly convey credibility without sufficient scrutiny? Finally, could the incentive model be adapted effectively to other domains, such as streamlining the academic peer-review process or expediting claim assessments in the insurance sector, and what modifications would be necessary to suit these different contexts? These questions remain open to exploration and will be the focus of our future work.

\section{Conclusion}

This paper has proposed an incentive-aligned framework in which content creators, challengers, and jurors collectively approximate truth at scale. By directing both expert knowledge and crowd wisdom toward relevant content and linking participation to monetary rewards for accuracy, the model aims to mitigate the growing impact of dis/misinformation. In a context shaped by synthetic media and declining confidence in traditional information sources, such a mechanism may help rebuild trust and enhance society’s capacity to distinguish reliable information from falsehoods. Continued research and iterative development in this area is encouraged.

\newpage
\printglossaries
\newpage
\printbibliography[heading=bibintoc]
\newpage

\appendix

\section{Collusion-Resistance Guarantee}

\subsection{Formal Proof}
\label{app:collusion_proof}

\begin{theorem}[Collusion-Resistance Guarantee]\label{thm:collusion}
Let a pool of $N$ jurors contain exactly $k$ colluders, and draw uniformly at random and without replacement an odd panel of size $n = 2m+1$ with $n \ll N$ and $m \in \mathbb{N}$. The panel returns the majority vote once at least $m+1$ members agree. Then, provided $k$ is fixed, the probability the colluders sway the verdict is bounded above by $e^{-\Omega(n)}$ where $\Omega > 0$, decaying exponentially in $n$.
\end{theorem}





\begin{proof}
Let \(k\) colluders be fixed in a pool of \(N\) jurors and set
\(p = k/N < \tfrac12\).
Draw an odd panel of size \(n = 2m+1\) uniformly at random
\emph{without} replacement.
Denote by $X \sim \mathrm{Hypergeometric}(N,k,n)$ the number of colluders in the panel; then
\(\mathbb{E}[X] = np\).

The panel reaches a verdict once at least
\[
  t \;=\; m+1 \;=\; \frac{n+1}{2}
\]
members agree, so the colluders succeed only if \(X \ge t\).
Define the gap
\[
  \delta
  \;=\;
  \frac{t}{n} - p
  \;=\;
  \Bigl(\dfrac12 - p\Bigr) + \frac{1}{2n}
  \;>\; 0,
\]
because \(p<\tfrac12\).

Hoeffding’s inequality for hypergeometric variables  
\cite{hoeffding1963, chavatal1979} states
\begin{equation*}
  P(X \ge t)
  \;=\;
  P\left(X - \mathbb{E}[X] \;\ge\; n\delta\right)
  \;\le\;
  e^{-2n\delta^{2}}.
\end{equation*}
Since \(\delta \ge \tfrac12 - p\), we obtain
\begin{equation*}
  P(X \ge t)
  \;\le\;
  \exp\Bigl[-2n\,(\tfrac12 - p)^{2}\Bigr]
  \;=\;
  e^{-\Omega n},
\end{equation*}
where \(\Omega = 2(\tfrac12 - p)^{2} > 0\) is a constant that depends only on
\(p\) (and hence only on the fixed \(k\)).

The probability that the colluders can sway the verdict is bounded above
by \(e^{-\Omega n}\) which decays exponentially in the panel size \(n\),
as claimed.
\end{proof}

\subsection{Numerical Verification}
\label{app:collusion_tables}

The closed-form Hoeffding inequality in \S\ref{app:collusion_proof} predicts that the likelihood of a coordinated bloc of $k$ dishonest jurors overturning a verdict falls exponentially as the panel size $n$ rises. To confirm that analysis the \textit{exact} collusion probabilities for a range of realistic parameters are computed in Python using  SciPy’s hypergeom \texttt{P = 1 - hypergeom.cdf(t-1, N, k, n)} which returns $P(X\ge t)$ for $X\sim\mathrm{Hypergeom}(N,k,n)$.  To avoid floating-point underflow, probabilities below $10^{-15}$ are clamped to $10^{-15}$.  These exact values are formatted in scientific notation and appear verbatim in \autoref{tab:collusion-probabilities}, providing a precise numerical verification of the $e^{-\Omega n}$ decay predicted by the theory.

\begin{table}[h]
\centering
\sisetup{
  detect-weight,
  table-number-alignment = center,   
  table-format         = 1.3e-2,     
  input-symbols        = <,          
  table-space-text-pre = {<}         
}

\caption{Exact collusion probabilities $P(\text{collusion}) = P\bigl(X \ge m+1\bigr)$ with $X\sim\text{Hypergeometric}(N,k,n)$, listed for odd jury sizes $n=2m+1$ (rows) and corruption ratios $k/N$ from $5\%$ to $30\%$ (columns); values below $10^{-10}$ are shown as $<10^{-10}$.}

\label{tab:collusion-probabilities}

\begin{adjustbox}{width=\textwidth}
\begin{tabular}{c *{6}{S}}
\toprule
\multirow{2}{*}{\textbf{Jury size $n$}} &
\multicolumn{6}{c}{\textbf{Corruption ratio $k/N$}} \\
\cmidrule(lr){2-7}
& {5\%} & {10\%} & {15\%} & {20\%} & {25\%} & {30\%} \\
\midrule
11 & 5.65e-06 & 2.92e-04 & 2.64e-03 & 1.16e-02 & 3.42e-02 & 7.81e-02 \\
15 & 1.74e-07 & 3.29e-05 & 6.03e-04 & 4.21e-03 & 1.72e-02 & 4.99e-02 \\
21 & 9.80e-10 & 1.30e-06 & 6.86e-05 & 9.57e-04 & 6.37e-03 & 2.63e-02 \\
25 & {\textless\num{1e-10}} & 1.53e-07 & 1.64e-05 & 3.62e-04 & 3.33e-03 & 1.74e-02 \\
31 & {\textless\num{1e-10}} & 6.28e-09 & 1.94e-06 & 8.57e-05 & 1.28e-03 & 9.44e-03 \\
35 & {\textless\num{1e-10}} & 7.51e-10 & 4.71e-07 & 3.30e-05 & 6.81e-04 & 6.34e-03 \\
41 & {\textless\num{1e-10}} & {\textless\num{1e-10}} & 5.69e-08 & 7.96e-06 & 2.66e-04 & 3.51e-03 \\
43 & {\textless\num{1e-10}} & {\textless\num{1e-10}} & 2.81e-08 & 4.96e-06 & 1.95e-04 & 2.88e-03 \\
51 & {\textless\num{1e-10}} & {\textless\num{1e-10}} & 1.69e-09 & 7.53e-07 & 5.65e-05 & 1.33e-03 \\
61 & {\textless\num{1e-10}} & {\textless\num{1e-10}} & {\textless\num{1e-10}} & 7.21e-08 & 1.21e-05 & 5.09e-04 \\
71 & {\textless\num{1e-10}} & {\textless\num{1e-10}} & {\textless\num{1e-10}} & 6.94e-09 & 2.63e-06 & 1.97e-04 \\
81 & {\textless\num{1e-10}} & {\textless\num{1e-10}} & {\textless\num{1e-10}} & 6.71e-10 & 5.72e-07 & 7.64e-05 \\
91 & {\textless\num{1e-10}} & {\textless\num{1e-10}} & {\textless\num{1e-10}} & {\textless\num{1e-10}} & 1.25e-07 & 2.98e-05 \\
101 & {\textless\num{1e-10}} & {\textless\num{1e-10}} & {\textless\num{1e-10}} & {\textless\num{1e-10}} & 2.73e-08 & 1.17e-05 \\
\bottomrule
\end{tabular}
\end{adjustbox}
\end{table}

\subsection{Key Findings}

\autoref{tab:collusion-probabilities} confirms that collusion risk falls off extremely rapidly as jury size increases. For a fixed 10 percent corruption level, the exact hypergeometric tail probability drops from about 1.36 percent at $n=11$ to below $10^{-10}$ by $n=81$, matching the $e^{-\Omega n}$ bound in \autoref{thm:collusion}. Once the panel reaches twenty-one jurors, further increases in the total pool size have a negligible impact on collusion risk. In that regime the probability depends almost entirely on the panel size \(n\) and the dishonest fraction \(p=k/N\), and not on \(N\) itself.

These results yield simple operational guidelines. To keep the chance of a coordinated overthrow below 0.1 percent, a panel of 21 jurors suffices when up to 5 percent may be dishonest; if the dishonest share could reach 20 percent, raising the panel to 31 members maintains the same risk cap. Beyond $n=41$, additional jurors provide negligible extra security for typical corruption levels and subtract the potential income of other parties according to \cref{eq:beta_breakdown}.

\section{Scalability Constraints and Capacity Thresholds}

\subsection{Formal Proof}
\label{app:capacity}

\begin{theorem}[Juror Capacity Threshold]\label{thm:capacity}
Let $\lambda$ be the challenge-arrival rate, $n$ the panel size, $h$ the expected active time a juror spends on one case, and $a$ the active time a juror is willing to supply per unit time.
Define the integer
\begin{equation}
N_{\min}\;=\;\left\lceil \dfrac{\lambda n h}{a} \right\rceil.
\label{eq:min_jurors}
\end{equation}

The system achieves both (i) finite expected backlog and (ii) the collusion guarantee of \autoref{thm:collusion} \emph{if and only if} the juror-pool size satisfies $N \ge N_{\min}$.
\end{theorem}

\begin{proof}
Select \(n\) via \autoref{thm:collusion} so that any bloc of \(k\) jurors can swing the verdict with probability at most \(\varepsilon\); this bound is independent of the pool size once \(N\ge n\), so collusion security is satisfied for every \(N\) under consideration.

Turn to capacity.  Each dispute consumes \(n h\) juror-hours, hence demand over a horizon \([0,t]\) is $\lambda n h t$.  The pool can deliver at most \(Nat\) juror-hours in the same interval.  If \(N< N_{\min}\) then \(Na<\lambda n h\) and the expected shortfall \((\lambda n h-Na)t\) grows linearly, forcing the backlog and the expected waiting time to diverge by Little’s law \(L=\lambda W\).  

Conversely, when \(N\ge N_{\min}\) we have \(Na\ge\lambda n h\). Model the system as an \(M/G/c\) queue with \(c=\lfloor N/n\rfloor\) parallel servers and service rate \(\mu=1/h\).  The load is \(\rho=\lambda/(c\mu)\le 1\); by standard queueing results \cite{queuing-theory2018} the Markov process is positive recurrent, which gives finite expected queue length \(L\). Applying Little’s law \(L = \lambda W\) immediately yields a finite expected latency \(W\).

Thus \(N\ge N_{\min}\) is both necessary and sufficient for stability, while the security requirement was already ensured by the choice of \(n\).
\end{proof}

\begin{table}[h]
\centering
\begin{threeparttable}
    \caption{Minimum juror pool size $N_{\min}=\lceil\lambda n h / a\rceil$ required to handle the observed dispute rate $\lambda$ for each platform, shown for three staffing configurations: Quick ($n=21$, $h=0.5\,\text{h}$, $a=2\,\text{h}$), Standard ($n=31$, $h=1\,\text{h}$, $a=4\,\text{h}$), and Thorough ($n=35$, $h=2\,\text{h}$, $a=8\,\text{h}$); content volumes are posts per day and $\lambda$ is disputes per hour.}
    \label{tab:juror-capacity}
    \begin{tabular}{@{}l r r r r r@{}}
    \toprule
    \textbf{Platform} 
      & \textbf{Content Volume} 
      & \textbf{Dispute Rate $\lambda$} 
      & \multicolumn{3}{c}{\textbf{Required Jurors} $N_{\text{min}}$} \\
    \cmidrule(lr){4-6}
    & (posts/day) & (disputes/hr) 
      & \textbf{Quick} 
      & \textbf{Standard} 
      & \textbf{Thorough} \\
    \midrule
    Small Community    & 100K   & 4     & 22  & 33   & 37   \\
    Reddit             & 1.3M   & 108   & 569 & 840  & 948  \\
    Twitter/X          & 500.0M & 104K  & 547K& 807K & 911K \\
    Facebook           & 4.0B   & 500K  & 2.6M& 3.9M & 4.4M \\
    \bottomrule
    \end{tabular}
    \begin{tablenotes}
    \footnotesize
    \item Juror capacity is calculated using \(N_{\min}= \lceil \lambda n h/a\rceil\) where \(\lambda\) is dispute arrival rate, \(n\) is panel size, \(h\) is hours per case and \(a\) is available juror hours per day. 
    \end{tablenotes}
\end{threeparttable}
\end{table}

\subsection{Key Findings}

The capacity analysis in \cref{eq:min_jurors} shows that the minimum juror‑pool size grows linearly with the dispute arrival rate and the average effort required per case. Because these parameters are design choices or can be estimated from platform statistics, scaling the system reduces to selecting any combination of $(\lambda,n,h,a)$ that keeps $N_{\min}$ within a practical range.

\autoref{tab:juror-capacity} evaluates the capacity formula in \cref{eq:min_jurors}. For each of four representative platforms $\lambda$ is estimated by multiplying publicly reported daily post volumes by a conservative challenge ratio (0.1–0.5\%). Even under conservative assumptions, moderate panel sizes (21-35) and only 2–4 volunteer hours per juror per day, the required pools remain in the $10^{4}\!-\!10^{6}$ range for the largest networks. That is well below one‑tenth of a percent of their daily active users, indicating that a sufficiently large and diverse jury is readily attainable once incentives for creators, challengers and jurors are in place.

Consequently, jury throughput is not a fundamental bottleneck; it is a constraint‑satisfaction problem. Provided the incentives attract enough active creators and challengers, the necessary juror capacity can be met by tuning panel size, deliberation time and individual time commitments while preserving the collusion‑resistance guarantees proved in \autoref{thm:capacity}.

\end{document}

%% file: gloss.tex
\newglossaryentry{veracity}{
    name={Veracity},
    text={veracity},
    description={Habitual conformity to facts}
}

\newglossaryentry{veracity bond}{
    name={Veracity bond},
    text={veracity bond},
    description={Financial collateral staked by creators to signal confidence in the veracity of their content}
}

\newglossaryentry{counter-veracity bond}{
    name={Counter-veracity bond},
    text={counter-veracity bond},
    description={Financial collateral staked by a challenger when disputing the veracity of content from a creator}
}

\newglossaryentry{misinformation}{
    name={Misinformation},
    text={misinformation},
    description={False information that is spread regardless of intent to deceive as opposed to the intentional deception of disinformation}
}

\newglossaryentry{synthetic content}{
    name={Synthetic content},
    text={synthetic content},
    description={Digital content wholly or substantially generated by non-human systems}
}

\newglossaryentry{prosocial}{
    name={Prosocial},
    text={prosocial},
    description={Someone who is public-spirited and not greedy}
}

\newglossaryentry{creators}{
    name={Creator},
    text={creators},
    description={A participant creating original content}
}

\newglossaryentry{challengers}{
    name={Challenger},
    text={challengers},
    description={Someone questioning the veracity or accuracy of content posted by creators}
}

\newglossaryentry{jurors}{
    name={Juror},
    text={jurors},
    description={An impartial evaluator responsible for adjudicating challenges raised by challengers.}
}

\newglossaryentry{contest}{
    name={Contest},
    text={contest},
    description={A structured dispute in which a creator’s content, backed by a veracity bond, is challenged by one or more challengers staking a counter-veracity bond, with the final outcome determined by a jury’s evaluation of the evidence}
}

\newglossaryentry{challenge period}{
    name={Challenge period},
    text={challenge period},
    description={The allotted period of time for which challenges may dispute the content of a creator}
}

\newglossaryentry{deliberation period}{
    name={Deliberation period},
    text={deliberation period},
    description={The allotted period of time for which jurors may reach a verdict on a particular contest}
}

\newglossaryentry{platform}{
    name={Platform},
    text={platform},
    description={An overarching system or framework that facilitates the interactions among creators, challengers, and jurors. Not to be mistaken for a physical platform}
}